\documentclass[a4paper,english,cleveref,autoref,thm-restate,nolineno]{socg-lipics-v2021}
\usepackage[utf8]{inputenc}

\usepackage{color}
\usepackage{csquotes}
\usepackage{xspace}
\usepackage{amsmath}
\usepackage{amssymb}
\usepackage{mathtools}

\newcommand{\norm}[1]{\ensuremath{\lVert #1 \rVert}}
\newcommand{\Oh}{\ensuremath{\mathcal{O}}\xspace}
\newcommand{\Ohtilda}{\ensuremath{\widetilde{\mathcal{O}}}\xspace}

\newcommand{\RR}{\mathbb{R}\xspace}
\newcommand{\NN}{\mathbb{N}\xspace}
\newcommand{\ZZ}{\mathbb{Z}\xspace}

\newcommand{\dtw}{\ensuremath{d_\text{DTW}\xspace}}
\newcommand{\dtwt}{\ensuremath{d_\text{DTW}^T\xspace}}
\newcommand{\gdtw}{\ensuremath{G_\text{DTW}\xspace}}
\newcommand{\taustart}{\ensuremath{\tau_\text{start}\xspace}}
\renewcommand{\epsilon}{\varepsilon}
\newcommand{\eps}{\varepsilon}
\DeclareMathOperator*{\argmin}{argmin}

\newcommand{\tOh}{\Ohtilda}

\hideLIPIcs

\title{Dynamic Time Warping Under Translation: Approximation Guided by Space-Filling Curves}

\titlerunning{(Approximation) Algorithms for Dynamic Time Warping Under Translation}

\author{Karl Bringmann}{Saarland University and Max Planck Institute for Informatics, Saarland Informatics Campus, Saarbrücken, Germany}{bringmann@cs.uni-saarland.de}{}{This work is part of the project TIPEA that has received funding from the European Research Council (ERC) under the European Unions Horizon 2020 research and innovation programme (grant agreement No. 850979).}
\author{Sándor Kisfaludi‑Bak}{Aalto University, Espoo, Finland}{sandor.kisfaludi-bak@aalto.fi}{}{Part of this research was conducted while the author was at the Max Planck Institute for Informatics, and part of it while he was at the Institute for Theoretical Studies, ETH Zürich.}
\author{Marvin Künnemann}{Institute for Theoretical Studies, ETH Zürich, Switzerland}{}{}{Research supported by Dr. Max Rössler, by the Walter Haefner Foundation,
and by the ETH Zürich Foundation.}
\author{Dániel Marx}{CISPA Helmholtz Center for Information Security, Saarbrücken, Germany}{marx@cispa.de}{https://orcid.org/0000-0002-5686-8314}{Research supported by the European Research Council (ERC) consolidator grant No.~725978 SYSTEMATICGRAPH.}
\author{André Nusser}{BARC, University of Copenhagen, Denmark}{anusser@mpi-inf.mpg.de}{https://orcid.org/0000-0002-6349-869X}{Part of this research was conducted while the author was at Saarbrücken Graduate School of Computer Science and Max Planck Institute for Informatics. The author is supported by the VILLUM Foundation grant 16582.}

\authorrunning{
K. Bringmann,
S. Kisfaludi‑Bak,
M. Künnemann,
D. Marx,
and A. Nusser} 

\Copyright{
Karl Bringmann,
Sándor Kisfaludi‑Bak,
Marvin Künnemann,
Dániel Marx,\\
and André Nusser} 



\category{} 

\relatedversion{}

\acknowledgements{We thank Christian Wulff-Nilsen for making us aware of an improved offline dynamic shortest paths data structure.}

\EventEditors{Xavier Goaoc and Michael Kerber}
\EventNoEds{2}
\EventLongTitle{38th International Symposium on Computational Geometry (SoCG 2022)}
\EventShortTitle{SoCG 2022}
\EventAcronym{SoCG}
\EventYear{2022}
\EventDate{June 7--10, 2022}
\EventLocation{Berlin, Germany}
\EventLogo{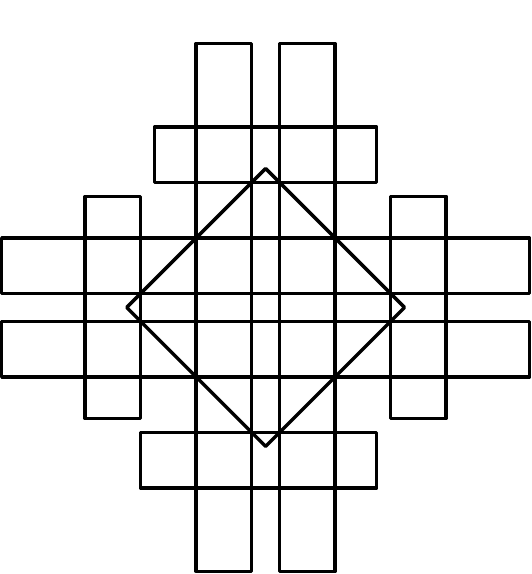}
\SeriesVolume{224}
\ArticleNo{XX}

\begin{document}

\maketitle

\begin{abstract}
	The Dynamic Time Warping (DTW) distance is a popular measure of similarity for a variety of sequence data. For comparing polygonal curves $\pi, \sigma$ in $\mathbb{R}^d$, it provides a robust, outlier-insensitive alternative to the Fréchet distance.
However, like the Fréchet distance, the DTW distance is not invariant under translations.
Can we efficiently optimize the DTW distance of $\pi$ and $\sigma$ under arbitrary translations, to compare the curves' \emph{shape} irrespective of their absolute location?  

	There are surprisingly few works in this direction, which may be due to its computational intricacy:
For the Euclidean norm, this problem contains as a special case the geometric median problem, which provably admits no exact algebraic algorithm (that is, no algorithm using only addition, multiplication, and $k$-th roots). We thus investigate exact algorithms for non-Euclidean norms as well as approximation algorithms for the Euclidean norm.

For the $L_1$ norm in $\mathbb{R}^d$, we provide an $\mathcal{O}(n^{2(d+1)})$-time algorithm, i.e., an exact polynomial-time algorithm for constant $d$. Here and below, $n$ bounds the curves' complexities.
For the Euclidean norm in $\mathbb{R}^2$, we show that a simple problem-specific insight leads to a $(1+\varepsilon)$-approximation in time $\mathcal{O}(n^3/\varepsilon^2)$. We then show how to obtain a subcubic $\widetilde{\mathcal{O}}(n^{2.5}/\varepsilon^2)$ time algorithm with significant new ideas; this time comes close to the well-known quadratic time barrier for computing DTW for fixed translations. Technically, the algorithm is obtained by speeding up repeated DTW distance estimations using a dynamic data structure for maintaining shortest paths in weighted planar digraphs. Crucially, we show how to traverse a candidate set of translations using space-filling curves in a way that incurs only few updates to the data structure. 

We hope that our results will facilitate the use of DTW under translation both in theory and practice, and inspire similar algorithmic approaches for related geometric optimization problems.
\end{abstract}

\section{Introduction} \label{sec:introduction}

Fast algorithms for computing similarity measures for sequence data enable a number of applications such as signature/handwriting recognition~\cite{yu_qiao_affine_2006, efrat_curve_2007}, map matching~\cite{BrakatsoulasPSW05, WenkSP06}, analysis of GPS tracking data~\cite{BrankovicBKNPW20} and many more. For polygonal curves in $\mathbb{R}^d$, a popular measure is the Fréchet distance~\cite{AltG95, EiterM94} -- we refer to~\cite{HarPeled_frechetChapter} for an overview over the extensive literature. Unfortunately, the Fréchet distance is very sensitive to outliers, as the distance value may easily be dominated by erroneous samplings of the curves. Consequently, some contexts would profit from a measure that is more robust to outliers, such as the average/integral Fréchet distance (see~\cite{Buchin07,MaheshwariSS18}) or the well-known dynamic time warping (DTW) distance. The DTW distance is particularly popular for audio sequences (such as speech recognition) and other domains, but has seen an increasing number of uses for geometric curves~\cite{DBLP:conf/iccv/MunichP99, yu_qiao_affine_2006, efrat_curve_2007, YingPFA16, BrankovicBKNPW20}.

Given two polygonal curves $\pi=(\pi_1, \dots, \pi_n)$ and $\sigma = (\sigma_1, \dots, \sigma_m)$ in $\mathbb{R}^d$, their DTW distance $\dtw(\pi, \sigma)$ can be defined as follows: We imagine a dog walking on $\pi$ and its owner walking on $\sigma$. Both owner and dog start at the beginning of their curves, and in each step independently decide to either stay in place or jump to the next vertex, until both of them have reached the end of their curves. Formally, this yields a traversal $T=((i_1,j_1), \dots, (i_t, j_t))$ where $i_1=j_1=1$, $i_t=n$, $j_t=m$ and $(i_{\ell+1}, j_{\ell+1})\in \{(i_\ell + 1, j_\ell), (i_\ell, j_\ell+1), (i_\ell + 1, j_\ell + 1)\}$. We define the cost of this traversal as the sum of distances of dog and owner during the traversal, i.e., $\sum_{\ell=1}^t \norm{\pi_{i_\ell} - \sigma_{j_\ell}}$. The corresponding DTW distance $\dtw(\pi, \sigma)$ is defined as the minimum cost of such a traversal.\footnote{For comparison, to obtain the discrete Fréchet distance of $\pi$ and $\sigma$, we would minimize, over all traversals~$T$, the \emph{maximum} distance of the dog and its owner during $T$ -- one may think of the smallest leash length required to connect dog and owner while traversing their curves.} Note that this measure depends on the metric space we use for our curves $\pi, \sigma$. For any metric that we can evaluate in constant time, a simple dynamic programming approach computes $\dtw(\pi, \sigma)$ in time $O(nm)$, i.e., time $O(n^2)$ when both curves have at most $n$ vertices. While one can achieve mild improvements over this running time~\cite{GoldS18}, one can rule out $O(n^{2-\epsilon})$-time algorithms under the Strong Exponential Time Hypothesis, already for curves in $\mathbb{R}$~\cite{AbboudBW15, BringmannK15}. Even for constant-factor approximations, no strongly subquadratic algorithms are known, see~\cite{Kuszmaul19} for (sub)polynomial approximation guarantees and~\cite{AgarwalFPY16, YingPFA16} for approximation algorithms on restricted input models. 

Unfortunately, the DTW distance is not \emph{translation-invariant}: Distant copies of the same curve may have a much larger distance than differently shaped curves that stay close to each other, see Figure~\ref{fig:dtw_vs_shape}. For certain curve similarity applications such as signature recognition, it is thus frequently argued (sometimes implicitly) that a translation-invariant measure is desirable, see e.g.~\cite{DBLP:conf/iccv/MunichP99, EfratIV01, vlachos_elastic_2005, yu_qiao_affine_2006,BergC11}.

Arguably the most natural way to make any curve distance measure translation-invariant is to take its minimum under translations of the curves: correspondingly, DTW under translation is defined as $\dtwt(\pi, \sigma) \coloneqq \min_{\tau \in \mathbb{R}^d} \dtw(\pi, \sigma + \tau)$. Unfortunately, for computing this translation-invariant measure, much less is known than, e.g., for the Fréchet distance under translation. This state of the art, which we review below, is the starting point for our work.

\begin{figure}
    \includegraphics{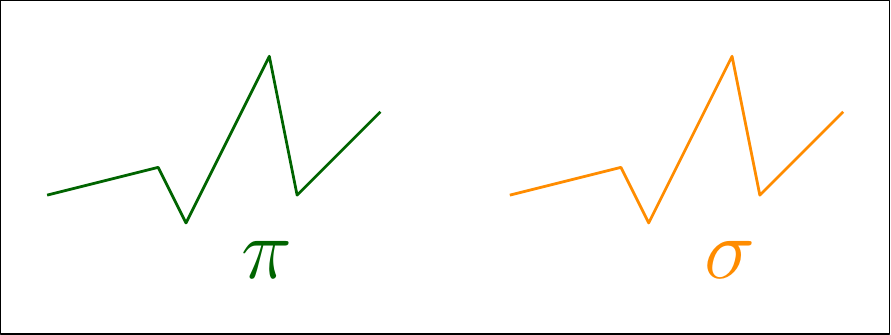}
	\hfill
    \includegraphics{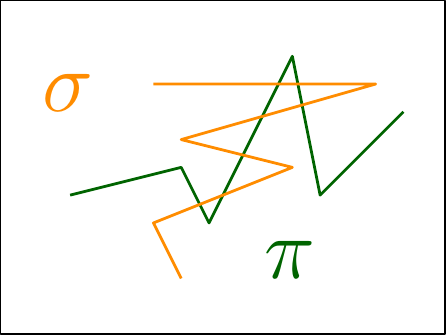}
	\caption{Curves with similar shape but large DTW distance (left) and different shape but small DTW distance (right).}
	\label{fig:dtw_vs_shape}
\end{figure}

\subparagraph*{Translation-invariant curve similarity measures.}
For the continuous Fréchet distance, the earliest algorithmic work studying its translation-invariant version dates back to 2001~\cite{EfratIV01, AltKW01}, with algorithms running in time $\tOh(n^{10})$ and $\tOh(n^8)$, respectively. For the discrete Fréchet distance under translation, algorithms have been improved from $\tOh(n^6)$~\cite{jiang2008protein}, via $\tOh(n^5)$~\cite{avraham2015faster}, to $\tOh(n^{4.667})$~\cite{BringmannKN21}, with a conditional lower bound of $n^{4-o(1)}$ based on the Strong Exponential Time Hypothesis~\cite{BringmannKN21}. These theoretical results have been complemented by an algorithm engineering study~\cite{BringmannKN20}. Approximation algorithms have been given by~\cite{EfratIV01, AltKW01}, including a $(1+\epsilon)$-approximation in time $O(n^2/\epsilon^2)$. Other works study related settings, such as more general transformations than translations~\cite{wenk2003phd, MosigC05}, or data structure variants~\cite{GudmundssonRSW21}. 

Unfortunately, we are not aware of algorithmic works with rigorous analyses for DTW under translation, but only heuristic approaches or works on related but different measures. Qiao and Yasuhara~\cite{yu_qiao_affine_2006} experimentally evaluate an iterative method for DTW distance under transformations including translation, rotation and scaling, but provide no theoretical guarantees. Vlachos, Kollios, and Gunopulos~\cite{vlachos_elastic_2005} study a closely related measure, a variation of the Longest Common Subsequence distance for geometric curves that is translation-invariant. This measure is similar to the DTW distance under translation using a binary distance metric with $d(x,y) = 0$ if $\norm{x-y}_\infty \le \epsilon$ and $d(x,y) = 1$ otherwise. For their measure, Vlachos et al.\ provide both exact and approximation algorithms. Munich and Perona~\cite{DBLP:conf/iccv/MunichP99} define another translation-invariant measure that roughly speaking minimizes differences in direction and velocity changes over traversals of the curves. Efrat, Fan, and Venkatasubramanian~\cite{efrat_curve_2007} study further variants of this measure.

One of the reasons for this lack of rigorous algorithmic work for DTW under translation may be its computational intricacy: Already when $\pi =(\pi_1,\dots,\pi_n)$ is a polygonal curve in $\mathbb{R}^d$ and $\sigma = (\sigma_1)$ consists of a single point $\sigma_1\in \mathbb{R}^d$, we obtain the geometric median problem as a special case. Specifically, the task simplifies to finding a point $x\in \mathbb{R}^d$ such that $\sum_{i=1}^n \norm{\pi_i - x}$ is minimized. This problem provably has no exact algebraic algorithm already for $n=5$ and $d=2$~\cite{Bajaj88} (that is, no algorithm using only
addition, multiplication, and $k$-th roots). We refer to~\cite{CohenLMPS16} for a recent near-linear time approximation algorithm and an overview of the literature on geometric median. By this lack of an exact, efficient algorithm for geometric median, we can thus hardly expect to solve DTW under translation in Euclidean spaces \emph{exactly}. This motivates to study the problem for norms other than Euclidean, as well as to study approximation algorithms for the Euclidean norm.

\subsection{Our results}

\subparagraph*{Exact algorithms for non-Euclidean norms.}
For the $L_1$ norm in $\mathbb{R}^d$, we give a polynomial-time exact algorithm whenever $d$ is constant.
\begin{theorem}
	For the $L_1$-norm in $\mathbb{R}^d$, we can solve DTW under translation  in time~$\Oh(n^{2(d+1)})$.
\end{theorem}
Since in $\mathbb{R}^2$ we can transform the $L_\infty$ norm to the $L_1$ norm by rotating the input by $\frac{\pi}{2}$ and scaling by $1/\sqrt{2}$, this also yields an $\Oh(n^{6})$ time algorithm for $L_\infty$ in $\mathbb{R}^2$.
We prove the result in Appendix~\ref{sec:exact_l1}.

\subparagraph*{Approximation algorithms for the Euclidean norm.}
The main focus in this paper is DTW under translation in the Euclidean plane. Since there is no exact algebraic algorithm due to the special case of geometric median, we focus on developing an \emph{approximation} algorithm.

As a first baseline, we observe that DTW under translation is at least as hard to compute as DTW for a fixed translation, even for approximation (we prove this in Appendix~\ref{sec:LBfixedtranslation}). Since exactly computing DTW for a fixed translation requires time $n^{2-o(1)}$ under the Strong Exponential Time Hypothesis~\cite{AbboudBW15, BringmannK15}, and no subquadratic-time constant-factor approximation algorithm is known, the best we could hope for with current techniques would be a $f(1/\epsilon) n^2$-time algorithm. Can we reach this baseline or does optimizing over translations in $\mathbb{R}^2$ increase the problem's complexity (and if so, by how much)? 
 
For the discrete Fréchet distance, optimizing over a translation increases the time complexity from $n^{2\pm o(1)}$~\cite{EiterM94, Bringmann14} to at least $n^{4-o(1)}$ and at most $O(n^{4.667})$~\cite{BringmannKN21} (where the lower bounds are based on the Strong Exponential Time Hypothesis). For $(1+\epsilon)$-approximations, a simple algorithm indeed manages to match the baseline of $O(n^2/\epsilon^2)$, see~\cite{AltKW01}. Does the same hold true for the DTW distance?

Similar arguments to~\cite{AltKW01} only achieve an $\tOh(n^4/\epsilon^2)$ time bound for DTW under translation. Using an insight specific to the nature of the DTW distance, we present a surprisingly simple $\tOh(n^3/\epsilon^2)$ time algorithm. We describe both approaches in Section~\ref{sec:tech_overview}. Our most important contribution is to obtain a \emph{subcubic} $\tOh(n^{2.5}/\epsilon^2)$ time bound via a sophisticated approach that exploits geometric arguments (specifically, a traversal via space-filling curves) to reduce our problem to maintaining shortest paths in a dynamically changing directed grid graph.
\begin{theorem}
	For the Euclidean norm in $\mathbb{R}^2$, we can solve $(1+\epsilon)$-approximate DTW under translation in time $\tOh(n^{2.5}/\epsilon^2)$.
\end{theorem}
Our techniques strengthen the paradigm of using dynamic algorithms for geometric optimization problems, for which we see a growing number of applications (besides classical examples such as~\cite{OvermarsY91}, see, e.g., recent work for the Fréchet distance under translation~\cite{avraham2015faster, BringmannKN21} or polygon placement~\cite{KunnemannN22}). Finally, only a sublinear factor of $\tOh(\sqrt{n})$ to the baseline of $\tOh(n^2/\epsilon^2)$ remains, which one might hope to decrease by further developing our ideas.

\subsection{Technical overview} \label{sec:tech_overview}

In this section, we describe the main ideas for our approximation algorithm for DTW under translation. To keep this exposition as simple as possible, we assume that both curves have the same complexity; let these curves be denoted by $\pi = (\pi_1,\dots, \pi_n)$ and $\sigma = (\sigma_1,\dots,\sigma_n)$ throughout this section.
The proof in Section~\ref{sec:apx_alg} gives the slightly more detailed arguments for possibly different complexities of the curves.
We start off with a simple algorithm that achieves a rather modest approximation guarantee: Let $\taustart \coloneqq \pi_1 - \sigma_1$ denote the translation of $\sigma$ that aligns the first points of $\pi$ and $\sigma + \tau$. It is straightforward to prove that the resulting DTW distance $\delta_\mathrm{start} \coloneqq \dtw(\pi, \sigma + \taustart)$ yields a $2n$-approximation to DTW under translation, i.e., $\dtwt(\pi, \sigma)\in [\delta_\mathrm{start}/(2n), \delta_\mathrm{start}]$. This follows from the fact that $\dtw(\pi, \sigma + \tau)$ is $(2n-1)$-Lipschitz with respect to $\tau$, and that $\tau^* \coloneqq \argmin_\tau \dtw(\pi, \sigma + \tau)$ satisfies $\norm{\taustart - \tau^*} \le \dtwt(\pi, \sigma)$, see Lemma~\ref{lem:2n-approx}. (Analogous arguments are known to give a 2-approximation for the Fréchet distance under translation~\cite{AltKW01, BringmannKN20}.)

With this rough approximation, the main task for approximating DTW under translation is to design an approximate decider with the following guarantee: Given the polygonal curves $\pi, \sigma$, a threshold $\delta > 0$ and approximation parameter $\epsilon > 0$, output a verdict ``$\dtwt(\pi, \sigma) \le (1+\epsilon)\delta$'' or ``$\dtwt(\pi, \sigma) > \delta$'' in time $T(n, \epsilon)$. In any case, the returned verdict has to be correct, i.e., if $\delta < \dtwt(\pi, \sigma)\le (1+\epsilon)\delta$ any output is admissible, otherwise it is uniquely determined. Given such an approximate decider, it is straightforward to obtain a $(1+\epsilon)$-approximation algorithm with running time $\Oh(T(n, \epsilon/3) \log (n/\epsilon))$ via binary search in the interval $[\delta_\mathrm{start}/(2n), \delta_\mathrm{start}]$, see Theorem~\ref{thm:main}.
We thus focus on the approximate decider for the remainder of this section.

\subparagraph*{A simple \boldmath$\Oh(n^4/\epsilon^2)$ solution.}
Let $B$ be the square of side length $2\delta$ centered at $\taustart = \pi_1 - \sigma_1$. To approximately decide whether $\dtw(\pi, \sigma) \le \delta$, we only need to consider translations in $B$, as any other translation~$\tau$ incurs a DTW distance larger than $\delta$ by $\dtw(\pi, \sigma + \tau) \ge \norm{\pi_1 - (\sigma_1 + \tau)} = \norm{\taustart - \tau} > \delta$.
Note that we can discretize this bounding box by a set~$Q$ of $\Oh( (n/\epsilon)^2 )$ translations such that for each translation $\tau^*\in B$, there is a close translation $\tau\in Q$ with $\norm{\tau^* - \tau} \le \frac{\epsilon \delta}{2n}$. Thus, if there is a translation $\tau^*$ with $\dtw(\pi, \sigma + \tau^*)\le \delta$, then by $(2n-1)$-Lipschitzness (Lemma~\ref{lem:lipschitz}), there is a translation $\tau \in Q$ with $\dtw(\pi, \sigma + \tau) \le \dtw(\pi, \sigma + \tau^*) + (2n-1)\cdot \norm{\tau^* - \tau} \le (1+\epsilon)\delta$. Consequently, by deciding $\dtw(\pi, \sigma + \tau) \leq (1+\eps)\delta$ for all $\tau\in Q$ using the exact $\Oh(n^2)$-time algorithm, we obtain an approximate decider with running time $T(n, \epsilon) = \Oh( n^4/\epsilon^2)$.

Note that the above arguments simplify the problem as follows: Find a set $Q$ of translations such that if there is some \emph{witness translation} $\tau^*$, i.e.,  $\dtw(\pi, \sigma+\tau^*) \le \delta$, then there is some $\tau \in Q$ with $\norm{\tau - \tau^*} \le \frac{\epsilon \delta}{2n}$. By computing $\dtw(\pi, \sigma + \tau)$ for all $\tau \in Q$, we can then approximately decide whether $\dtwt(\pi, \sigma) \le \delta$.

\subparagraph*{A more careful \boldmath$\Oh(n^3/\epsilon^2)$ solution.}
It turns out that we can significantly reduce the size of the set $Q$ by analyzing the properties of good DTW traversals more closely.
Consider a DTW traversal $((i_1, j_1), \dots, (i_t, j_t))$ of $\pi$ and $\sigma + \tau^*$, with traversal cost $\sum_{\ell = 1}^t \norm{\pi_{i_\ell} - (\sigma_{j_\ell} + \tau^*)} \le \delta$. 
Then, by a simple Markov argument, there can be at most $n/2$ pairs $\pi_{i_\ell}, \sigma_{j_\ell}$ with $\norm{\pi_{i_\ell} - (\sigma_{j_\ell} + \tau^*)} \ge 2\delta/n$, since otherwise already these pairs would lead to a traversal cost of more than $\delta$.
Since the traversal has $t\ge n$ steps, it follows that there are at least $t-n/2\ge n/2$ pairs $\pi_{i_\ell}, \sigma_{j_\ell}$ with $\norm{\pi_{i_\ell} - (\sigma_{j_\ell} + \tau^*)} \le 2\delta/n$.
Since $\pi_\ell - (\sigma_\ell + \tau^*) = (\pi_\ell - \sigma_\ell) - \tau^*$, this yields an important restriction on $\tau^*$: 
\begin{center} \itshape
	For any $\tau^*$ such that $\pi$ and $\sigma+\tau^*$ have DTW distance at most $\delta$,\\ there exist at least $n/2$ pairs $\pi_{i}, \sigma_{j}$ with $\norm{(\pi_i - \sigma_{j}) - \tau^*} \le 2\delta/n$.
\end{center}

This property immediately gives a simple randomized $\tOh(n^3/\epsilon^2)$ algorithm: Simply draw a pair $(i,j)$ uniformly at random from $[n]^2$ and test all translations $\tau$ given by $O(1/\epsilon^2)$ equally-spaced points in a $[-\frac{2\delta}{n}, \frac{2\delta}{n}]^2$-box $C_{i,j}$ centered at $\pi_i - \sigma_j$. If there is some $\tau^*\in C_{i,j}$ with $\dtw(\pi, \sigma+\tau^*) \le \delta$, one of the checked translations $\tau$ achieves $\dtw(\pi, \sigma + \tau)\le (1+\epsilon)\delta$. By the above property, we have that $\tau^*\in C_{i,j}$ with probability at least $(n/2)/n^2 = 1/(2n)$. Thus, it suffices to repeat this process $\tOh(n)$ times to find a good translation with high probability, if one exists. This yields a total running time of $\tOh(n \cdot \frac{1}{\eps^2} \cdot n^2) = \tOh(n^3/\epsilon^2)$.

In order to leverage this property \emph{deterministically}, define the multiset $P \coloneqq \{ \pi_i - \sigma_j \mid i,j\in [n]\}$ of $n^2$ points.
Recall that $B$ is the square of side length $2\delta$ centered at $\taustart = \pi_1 - \sigma_1$.
We impose a grid on the bounding box $B$ where each grid cell has side length $2\delta/n$. Consider a translation $\tau^*$ in some grid cell $C$ such that $\pi$ and $\sigma + \tau^*$ have DTW distance at most $\delta$. Then there must be $n/2$ points $p\in P$ with $\norm{p-\tau^*} \le 2\delta/n$ -- these points are distributed among $C$ and at most 3 neighboring cells of $C$.\footnote{Here, we say that two cells are neighboring if they share a common vertex.}
Thus, for any witness translation $\tau^*$, there must be a neighboring (including itself) grid cell containing at least $n/8$ points from $P$ -- we call such a cell \emph{dense}. Thus, we only need to check for translations that are inside a dense cell or neighboring a dense cell. Since $|P| = n^2$, there can be at most $|P|/(n/8)=8n$ dense cells, resulting in $\Oh(n)$ cells to check for a good translation.

Since each grid cell has side length $\Oh(\delta/n)$, we can discretize each relevant cell $C$ by $\Oh(1/\epsilon^2)$ many translations $Q_C$ such that if any $\tau^*\in C$ achieves $\dtw(\pi, \sigma+\tau^*) \le \delta$, then there is a $\tau\in Q_C$ with $\norm{\tau^*-\tau} \le \epsilon \delta / (2n)$ and thus $\dtw(\pi, \sigma + \tau) \le (1+\epsilon)\delta$. Thus, by letting $Q$ be the union of $Q_C$ for all $\Oh(n)$ cells $C$ that we need to check, we obtain $|Q|=\Oh(n/\epsilon^2)$, significantly improving over the previous bound of $\Oh(n^2/\epsilon^2)$. Computing the DTW distance for each translation in $Q$, we obtain a \emph{deterministic} $\Oh(n^3/\epsilon^2)$-time algorithm.

\begin{figure}
	\begin{center}
		\includegraphics[width=0.7\textwidth]{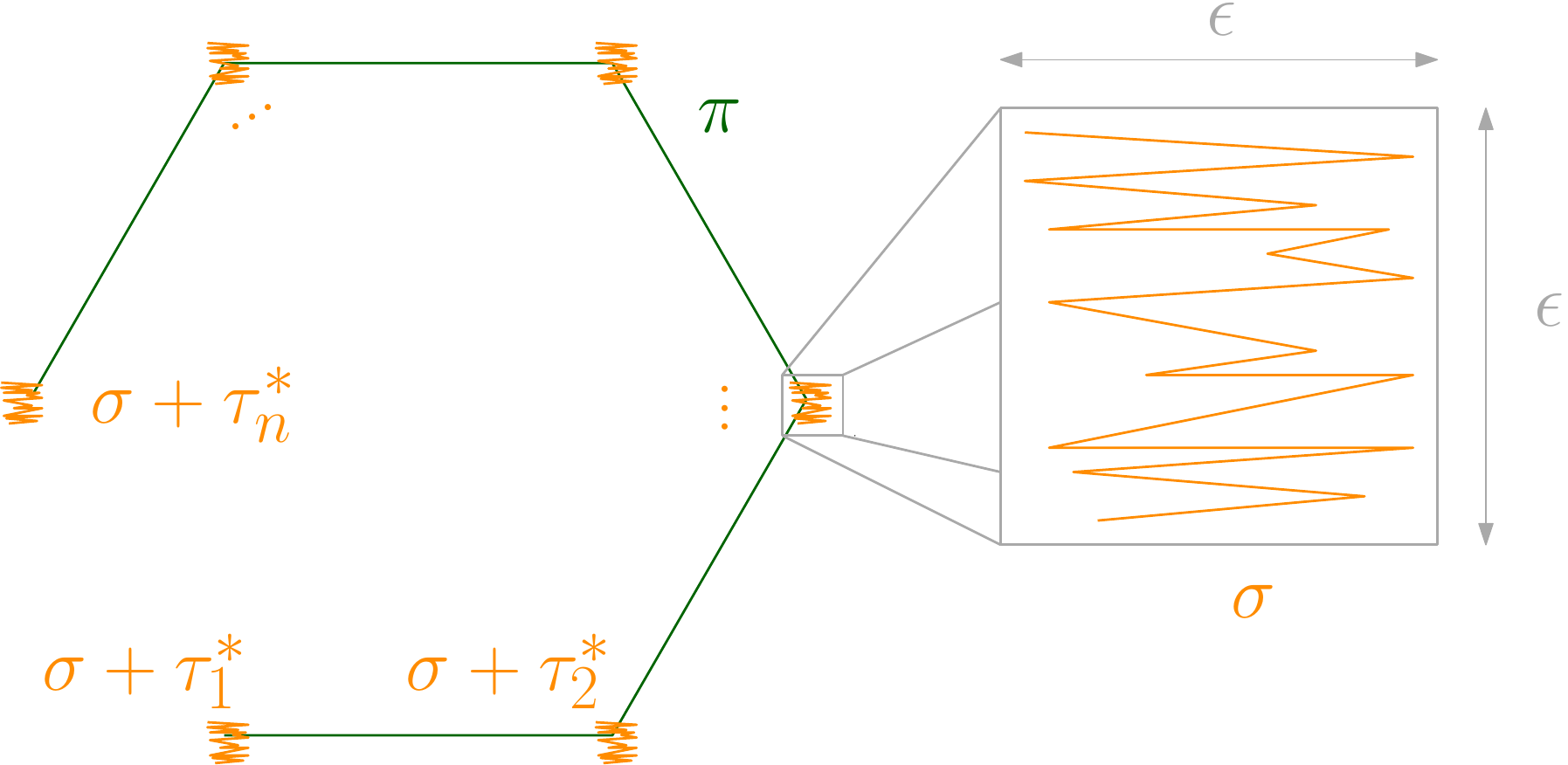}
	\end{center}
	\caption{If $\pi$ is given by a regular $n$-gon and $\sigma$ is a $3n$-vertex curve in a small $[0,\epsilon]^2$ area, where $\epsilon$ is small, then the DTW under translation distance can have $\Omega(n)$ local optima, each of which is near-optimal. These local optima correspond to translating $\sigma$ towards each vertex of $\pi$.}
	\label{fig:local-optima}
\end{figure}
\subparagraph*{Beating \boldmath$\Oh(n^3/\epsilon^2)$.}
Can we improve over the previous algorithm? A first idea would be to try to reduce the size of $Q$ even further, below $\Theta(n \cdot \mathrm{poly}(1/\eps))$. However, there is evidence that this route is rather difficult: One can construct instances with $\Omega(n)$ many near-optimal local optima that are well-separated from each other, see Figure~\ref{fig:local-optima}.
It thus appears quite challenging to avoid a check of $\Omega(n)$ regions of translations.

A different route is to speed up the computation of DTW distances $\dtw(\pi, \sigma + \tau)$ over all $\tau\in Q$, avoiding the naive time bound of $\Oh(|Q| \cdot n^2)$. Such approaches have been proven successful for related geometric optimization problems, such as Fréchet distance under translation~\cite{avraham2015faster, BringmannKN21} or polygon placement~\cite{KunnemannN22}. Crucially, one needs to exploit that the $|Q|$ distance computations are  related (for solving $|Q|$ independent instances, a conditional lower bound of $(|Q|n^2)^{1-o(1)}$ can be shown based on the quadratic-time hardness for DTW~\cite{AbboudBW15,BringmannK15}). To this end, we open up the black-box $\Oh(n^2)$-algorithm for DTW.

\begin{figure}
	\centering
	\includegraphics[width=\textwidth]{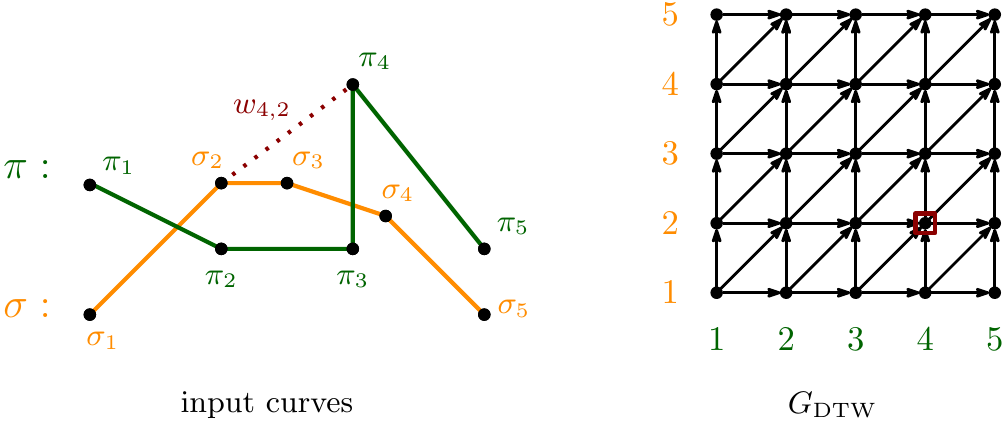}
	\caption{An example of two curves $\pi, \sigma$ and their dynamic time warping graph $\gdtw$.}
	\label{fig:dtw_graph}
\end{figure}

Given $\pi = (\pi_1,\dots, \pi_n)$ and $\sigma = (\sigma_1,\dots,\sigma_n)$, let $\gdtw$ denote the node-weighted directed grid graph with vertex set $V=\{(i,j)\mid i,j \in [n]\}$ and edge set $E$ consisting of horizontal edges from $(i,j)$ to $(i+1, j)$, vertical edges from $(i,j)$ to $(i,j+1)$ and diagonal edges from $(i,j)$ to $(i+1,j+1)$. Each node $(i,j)$ receives the weight $w_{i,j} = \norm{\pi_i - \sigma_j}$. Then it is not difficult to see that $\dtw(\pi, \sigma)$ is equal to the distance from $(1,1)$ to $(n,n)$ in $\gdtw$. As such, we can exploit algorithmic results on maintaining shortest paths in weighted planar digraphs under weight updates (here, one usually considers edge-weighted graphs, which subsumes the node-weighted setting). Unfortunately, when translating $\sigma$ by $\tau$, $\Omega(n^2)$ weights may change in $\gdtw$ so that even constant-time updates would lead to an $\Omega(n^3/\epsilon^2)$ time solution. In contrast, work on the Fréchet distance under translation~\cite{avraham2015faster,BringmannKN21} considers translations in an order that incurs only $\Oh(1)$ updates per translation.

Surprisingly, one can indeed reduce the number of weight updates below $\Oh(n^2)$ when we resort to \emph{approximating} each weight $w_{i,j}$ by an estimate $\norm{\pi_i - (\sigma_j + \tau)}/(1+\epsilon) \le w_{i,j} \le (1+\epsilon) \cdot \norm{\pi_i - (\sigma_j + \tau)}$. Specifically, we show how to traverse the $\Oh(n/\epsilon^2)$ translations in $Q$ in an order specified by a \emph{space-filling curve} such that we only need to update $\tOh(n^2/\epsilon^2)$ weights in total to maintain approximate weights. This statement and its analysis is one of the most interesting technical contributions of this paper and is proven in Section~\ref{sec:dynamic_graph_problem}. It remains to report the shortest distance from $(1,1)$ to $(n,n)$ in the directed grid graph $\gdtw$ for $\Oh(n/\epsilon^2)$ queries and $\tOh(n^2/\epsilon^2)$ weight updates.
For this task, we use the data structure due to Das et al.~\cite{DPGWN22} whose parameters can be set to give query time $\tOh(N^{3/4})$ and update time $\tOh(N^{1/4})$ for weighted planar digraphs with $N$ vertices.
Since $N=n^2$, we obtain a total running time of $\tOh(n^{2.5}/\epsilon^2)$, which improves polynomially over the previous $\tOh(n^3/\epsilon^2)$ solution.
We believe that our approach of maintaining approximate weights efficiently using a space-filling curve traversal may turn out useful for further improvements in similar contexts of geometric optimization problems.

\section{Preliminaries \& notation} \label{sec:prelims}

To denote index sets we use the notation $[n] \coloneqq \{1, \dots, n\}$.
Let $\pi = (\pi_1, \pi_2, \dots, \pi_n)$ and $\sigma = (\sigma_1, \sigma_2, \dots, \sigma_m)$ be two sequences of points in $\RR^d$. We assume $n \ge m$ without loss of generality.
To define the Dynamic Time Warping distance (DTW), we first introduce traversals. A sequence of index pairs $T = ((i_1,j_1), (i_2,j_2), \dots, (i_L, j_L))$ is a traversal of two curves of complexity $n$ and $m$ if $(i_1, j_1) = (1,1)$, $(i_L, j_L) = (n,m)$, and $(i_{\ell+1}, j_{\ell+1}) \in \{(i_\ell + 1, j_\ell), (i_\ell, j_\ell + 1), (i_\ell + 1, j_\ell + 1)\}$ for each $\ell \in [L-1]$. We call $L$ the number of steps of the traversal $T$.
Let $\mathcal{T}_{n,m}$ be the set of all traversals of curves of length $n$ and $m$.
The Dynamic Time Warping distance between $\pi$ and $\sigma$ is then defined as 
\[
	\dtw(\pi,\sigma) \coloneqq \min_{T \in \mathcal{T}} \sum_{(i,j) \in T} d(\pi_i, \sigma_j),\]
where for the metric $d(\cdot, \cdot)$, we use the $L_p$-norm $d(x,y)=\norm{x-y}_p$ throughout this paper. In the remainder, we omit the $p$ as it is either clear from the context, or the statement holds for all $p \in [1, \infty)$.
Furthermore, we often use bounds on the number of steps of the traversal. To that end, note that for $m \leq n$, any traversal in $\mathcal{T}_{n,m}$ consists of at least $n$ and at most $n+m-1$ steps.

For a sequence $\pi = (\pi_1,\dots,\pi_n)$ with $\pi_i \in \RR^d$ and a translation $\tau \in \RR^d$, we define the translated sequence as $\pi + \tau \coloneqq (\pi_1 + \tau, \pi_2 + \tau, \dots, \pi_n + \tau)$.
Dynamic Time Warping Under Translation is then defined as $\dtwt(\pi,\sigma) \coloneqq \min_{\tau \in \RR^d} \dtw(\pi, \sigma+\tau)$.
Recall that a function $f \colon \mathbb{R}^d \to \mathbb{R}$ is called $L$-Lipschitz (with respect to norm $\|.\|$) if for any $\tau, \tau' \in \mathbb{R}^d$ we have $|f(\tau) - f(\tau')| \le L \cdot \|\tau - \tau'\|$.
We prove the following lemma in Appendix~\ref{sec:deferred_proofs}.

\begin{lemma} \label{lem:lipschitz}
$\dtw(\pi,\sigma+\tau)$ is $(n+m-1)$-Lipschitz in $\tau$.
\end{lemma}

The following lemma gives a simple $(n+m)$-approximation for DTW under translation and is a straightforward adaption of a corresponding $2$-approximation for the Fréchet distance under translation~\cite[Observation 2]{BringmannKN20}.  Note that one can create simple examples where this approximation ratio is almost tight. Again, we defer the proof to Appendix~\ref{sec:deferred_proofs}.
\begin{lemma}\label{lem:2n-approx}
	Let $\taustart = \pi_1-\sigma_1$. Then $\dtw(\pi, \sigma + \taustart) \le (n+m) \cdot \dtwt(\pi, \sigma)$. 
\end{lemma}

As discussed in Section~\ref{sec:introduction}, DTW corresponds to a grid graph problem. We now formally define this.
Given a DTW instance with curves $\pi = (\pi_1, \dots, \pi_n)$ and $\sigma = (\sigma_1, \dots, \sigma_m)$, we define a directed graph $\gdtw = (V,E,w)$ on a node-weighted grid (including certain diagonals) with node set $V \coloneqq \{ (i,j) \mid i \in [n], j \in [m] \}$, 
edge set
\begin{equation*}
\begin{split}
	E \coloneqq \{ ((i,j), (i+1,j)) \mid i \in [n-1], j \in [m] \} \quad \cup \quad \{ ((i,j), (i,j+1)) \mid i \in [n], j \in [m-1] \} \\ \cup \quad \{ ((i,j), (i+1,j+1)) \mid i \in [n-1], j \in [m-1] \},
\end{split}
\end{equation*}
and weights $w: V \to \RR$ with $w((i,j)) \coloneqq \|\pi_i - \sigma_j\|$. To simplify notation, we write $w_{i,j}$ instead of $w((i,j))$ to denote the weight of node $(i,j)$.
Note that finding a shortest path in this graph from $(1,1)$ to $(n,m)$ is equivalent to finding the minimum cost traversal.

\medskip
In order to define the order of the updates and queries in the dynamic graph problem that we introduce in Section~\ref{sec:apx_alg}, we use a space-filling curve on a grid.
Let
\[
	\mathcal{G}_R \coloneqq \{ i \cdot R \mid i \in \ZZ \} \times \{ j \cdot R \mid j \in \ZZ \}.
\]
be an infinite grid with resolution $R \in \RR$.
For our purpose, a space-filling curve is a hierarchical traversal of a finite grid: we partition this grid into four parts and, in some fixed order of the parts, recursively traverse each subgrid exhaustively before traversing the next one.
More precisely, we define the curve on the $2^k \times 2^k$ grid $C^k_{0,0} \coloneqq \mathcal{G}_R \cap [0,(2^k-1)R]^2$ for some $R$, and we recursively split $C_{i,j}^\ell$ into the boxes $C_{2i,2j}^{\ell-1}, C_{2i,2j+1}^{\ell-1}, C_{2i+1,2j}^{\ell-1}, C_{2i+1,2j+1}^{\ell-1}$ until they only contain a single grid point, i.e., until $\ell = 0$. This leads to the following definition:
\[
	C_{i,j}^\ell \coloneqq \{ (i2^\ell + s)R \mid s \in \{0, \dots, 2^\ell-1\} \} \times \{ (j2^\ell + s)R \mid s \in \{0, \dots, 2^\ell-1\} \}.
\]
For each cell $C_{i,j}^\ell$ with $\ell > 0$, the space-filling curve then traverses the points of the children in a way such that for each child all points are traversed in a continuous piece. 
For example, we first traverse all points of $C_{2i,2j+1}^{\ell-1}$, then $C_{2i+1,2j+1}^{\ell-1}$, then $C_{2i,2j}^{\ell-1}$, and finally $C_{2i+1,2j}^{\ell-1}$.
Recursively applying this leads to a sequence $z_1, \dots, z_{2^{2k}}$ of all points in $\mathcal{G}_R \cap [0,(2^k-1)R]^2$.
This sequence is called the z-curve, see Figure~\ref{fig:zcurve}. (However, any other order to traverse the children also works for our purpose.)

\begin{figure}
	\centering
	\includegraphics[width=.8\textwidth]{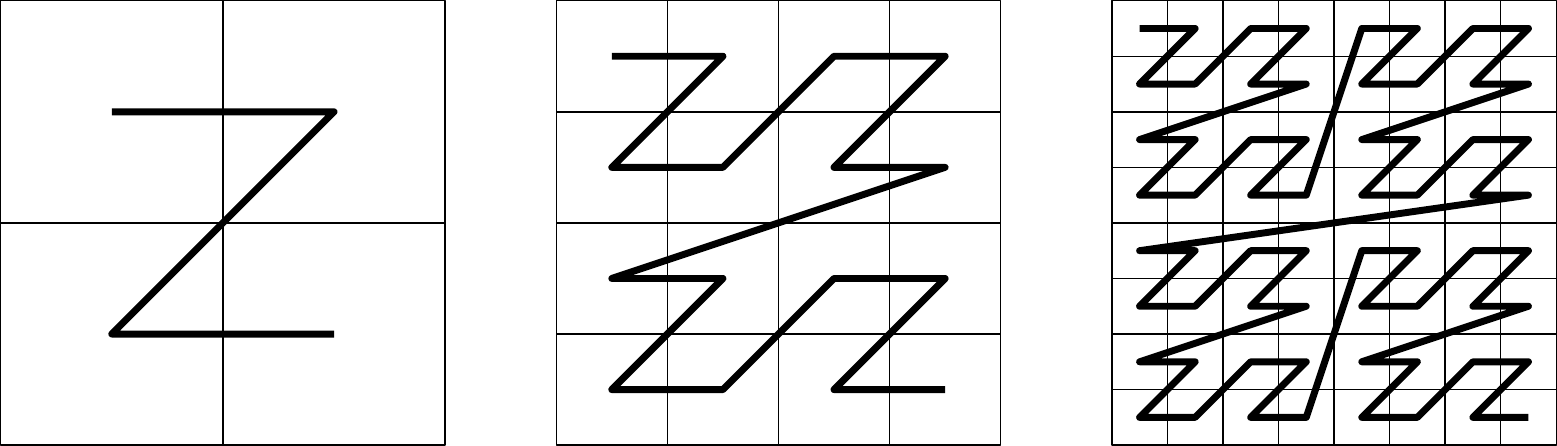}
	\caption{The traversal of the z-curve for $k = 1$, $k = 2$, and $k=3$.}
	\label{fig:zcurve}
\end{figure}

To argue about the space-filling curve traversals, it is sometimes useful to view the grid of points $\mathcal{G}_R$ equivalently as a grid of cells, i.e., as a set of squares partitioning $\mathbb{R}^2$. To switch between these views, build the Voronoi diagram of the point grid to obtain the cell grid, and conversely, use the center of each cell to obtain the point grid. We will freely use whichever view is most convenient in any context.

\section{\boldmath Approximating DTW under translation in $L_p$} \label{sec:apx_alg}

In this section we present an $\Ohtilda(n^{2.5}/\eps^2)$ algorithm for the problem of $(1+\eps)$-approximating DTW under translation in the Euclidean plane.
The algorithm that we present consists of two parts. First, we reduce to a dynamic shortest path problem on a grid graph. Second, we show that with the resulting number of updates and queries, we can use an existing dynamic graph algorithm to then obtain a subcubic algorithm for the problem at hand.

Recall that we consider the approximate decision problem: Given sequences $\pi = (\pi_1, \dots, \pi_n)$ and $\sigma = (\sigma_1, \dots, \sigma_m)$ with $\pi_i, \sigma_j \in \RR^2$, a distance $\delta \in \RR$, and an approximation parameter $\eps > 0$, either decide that $\dtwt(\pi, \sigma) \leq (1+\eps) \delta$ or that $\dtwt(\pi, \sigma) > \delta$. Recall that we assume $n \ge m$.
We first present a basic cubic algorithm that already captures some important properties of the subcubic algorithm that we subsequently present. 

\subsection{Cubic algorithm} \label{sec:cubic_alg}

We now present the cubic algorithm that was already outlined in Section~\ref{sec:tech_overview}.
First, if $\delta = 0$, we make a precise decision by testing for $\dtw(\pi, \sigma+\taustart) = 0$ with $\taustart = \pi_1 - \sigma_1$.
To facilitate the presentation, we furthermore assume that $\epsilon$ is given such that $\frac{n}{\epsilon} = 2^{k}$ for some $k \in \NN$. We can easily achieve this by rounding the input $\epsilon$ down to the largest value that fulfils this constraint, which changes the value of $\epsilon$ by at most a factor of 2.

In another preprocessing step, we round the coordinates of the points of $\pi$ and $\sigma$ to the closest multiple of $\frac{\delta}{4n} \eps$. This is feasible as it changes the DTW distance by less than
\[
	(n+m) \cdot \frac{\delta}{4n} \eps \leq \delta \frac{\eps}{2}.
\]
The multiset of translations from any point in $\sigma$ to any point in $\pi$ is then defined as
\[
	P \coloneqq \{ \pi_i - \sigma_j \mid i \in [n] \text{ and } j \in [m] \}.
\]
Note that by construction also all coordinates of all points in $P$ are multiples of $\frac{\delta}{4n} \eps$. Furthermore, as $P$ is a multiset, we have $|P| = nm$.
We now define a set of boxes that enables us to find dense regions. Consider the square $B \coloneqq [-\delta, \delta]^2 + \taustart$.
Partition $B$ into $n^2$ boxes $B_1, \dots, B_{n^2}$ of size $\frac{2\delta}{n} \times \frac{2\delta}{n}$ (note that their boundaries might intersect).
We now formally define the notion of a dense box already introduced intuitively in Section~\ref{sec:tech_overview}.
\begin{definition}[Dense Box]
	A box $B_i$ is dense if at least $\frac{n}{18}$ points of $P$ are contained in $B_i$.
\end{definition}
As $|P| = nm$, we obtain the following observation:
\begin{observation} \label{obs:dense_box}
	There are at most $18m$ dense boxes.
\end{observation}
Note that we can find the dense boxes in time $\Ohtilda(|P|)$ by associating each point with the tuple of indices in $[n] \times [n]$ of its containing box and then sorting these tuples.
Now, let $N(B_i)$ be the neighborhood of a box $B_i$, i.e., $N(B_i) \coloneqq \{ B_j \mid \text{$B_j \cap B_i \neq \emptyset$} \}$.
Note that $B_i \in N(B_i)$, so each box has (up to) 9 neighbors.
The crucial property of dense boxes is that any witness translation $\tau$ with $\dtw(\pi, \sigma + \tau) \leq \delta$ has to be in the neighborhood of a dense box:
\begin{lemma} \label{lem:dense_cell}
	If $\dtwt(\pi, \sigma) \leq \delta$, then there exists a dense box $B_j$, a neigbor $B_i \in N(B_j)$, and a $\tau \in B_i$ such that $\dtw(\pi, \sigma + \tau) \leq \delta$.
\end{lemma}
The proof is deferred to Appendix~\ref{sec:deferred_proofs}.

As we have to approximately decide whether $\dtw(\pi, \sigma+\tau) \leq \delta$ for any $\tau$ that is neighboring a dense box, we intersect each of these boxes with an $8\epsilon \times 8\epsilon$ grid and this gives us the set of points $Q$ that we have to evaluate. More precisely, let
$\mathcal{G} \coloneqq \mathcal{G}_{\frac{\delta}{4n}\eps} \cap B$,
where again $B = [-\delta, \delta]^2 + \taustart$.
Note that all points of $\mathcal{G}$ are still integer multiples of $\frac{\delta}{4n}\epsilon$. We now define our set of evaluation points to be
\[
	Q \coloneqq \{ \mathcal{G} \cap B_i \mid B_i \in N(B_j) \text{ for some dense box $B_j$} \}.
\]
Note that from Observation~\ref{obs:dense_box} and the bound $|\mathcal{G} \cap B_i| \in \Oh(\frac{1}{\eps^2})$, it follows that $|Q| \in \Oh(\frac{m}{\eps^2})$. 

Computing $\dtw(\pi,\sigma+q)$ for each $q \in Q$ suffices to implement an approximate decider. Indeed, if for some $q \in Q$ we find $\dtw(\pi,\sigma+q) \le (1+\eps)\delta$, then we conclude that $\dtwt(\pi,\sigma) \le (1+\eps)\delta$. Otherwise, if $\dtw(\pi, \sigma+q) > (1+\eps)\delta$ for all $q \in Q$, then we conclude that $\dtwt(\pi,\sigma) > \delta$, by the following correctness lemma proven in Appendix~\ref{sec:deferred_proofs}.
\begin{lemma}[Correctness] \label{lem:correctness}
If $\dtw(\pi, \sigma+q) > (1+\eps)\delta$ for all $q \in Q$, then $\dtwt(\pi,\sigma) > \delta$.
\end{lemma}

If we just evaluate each point in $Q$ naively, then the running time is $\Oh(n m^2 \left(\frac{1}{\eps}\right)^2)$, as there are $\Oh(m)$ dense cells, each of them with $\left(\frac{1}{\eps}\right)^2$ grid points, and each DTW evaluation takes time $\Oh(nm)$.
In the next section, instead of naively recomputing DTW for each translation, we dynamically update the DTW graph weights and then query for the shortest path.

\subsection{Reduction to dynamic graph problem} \label{sec:dynamic_graph_problem}

Now we present the first step in solving DTW under translation in subcubic time. To this end, we transform our problem into a dynamic shortest path problem on a grid graph.

\paragraph*{Dynamic graph problem.}
Recall that computing DTW for a fixed translation is a shortest path problem on a grid graph, see Section~\ref{sec:prelims}. More precisely, in the grid graph with node weights $w_{i,j} = \|\pi_i - (\sigma_j + q)\|$ the shortest path distance from $(1,1)$ to $(n,m)$ is equal to $\dtw(\pi,\sigma+q)$. However, as we only want to compute a $(1+\eps)$-approximation, we can relax the condition on the node weights to:
\begin{align} \label{eq:nodeweights}
	\frac{\|\pi_i - (\sigma_j+q)\|}{(1+\eps)} \leq w_{i,j} \leq (1+\eps)\|\pi_i - (\sigma_j+q)\|.
\end{align}
Observe that for such node weights the shortest path distance from $(1,1)$ to $(n,m)$ is equal to $\dtw(\pi,\sigma+q)$ up to a factor $(1+\eps)$. 

We choose the same set of query translations $Q$ as in Section~\ref{sec:cubic_alg}. We iterate over all $q \in Q$, and for each $q$ we first update the node weights in the grid graph in order to satisfy (\ref{eq:nodeweights}) and then we query the shortest path distance from $(1,1)$ to $(n,m)$ in the grid graph, obtaining a $(1+\eps)$-approximation of $\dtw(\pi, \sigma+q)$. As in Section~\ref{sec:cubic_alg}, this yields a $(1+\Oh(\eps))$-approximation of $\dtwt(\pi,\sigma)$ (and after scaling $\eps$ this becomes a $(1+\eps)$-approximation). 
Note that we did not fix the ordering of the query translations $q \in Q$ yet. In the following, we first fix this ordering, and then argue that our ordering guarantees that the total number of node weight updates is small, and furthermore we can efficiently determine which node weight updates have to be performed.

\paragraph*{Query ordering.}
Consider the z-curve over the grid $\mathcal{G} = \mathcal{G}_{\frac{\delta}{4n}\eps} \cap [-2\delta, 2\delta)^2 + \taustart$.
Note that this z-curve has depth $\log_2(16\frac{n}{\eps})$, and recall that $\frac{n}{\eps}$ is a power of 2.
The points of $Q$ lie on the z-curve, as $\pi, \sigma$ are rounded and the grid has resolution $\frac{\delta}{4n}\eps$. Thus, the z-curve induces an ordering of $Q$, and this is the ordering that we choose.

\paragraph*{Updates.}
We now describe how we determine the node weight updates in the dynamic grid graph problem, to ensure that when we run the shortest path query corresponding to $q \in Q$ the node weights satisfy (\ref{eq:nodeweights}).
We first argue why the total number of node weight updates is small, and subsequently discuss how to compute the sequence of node weight updates.
\begin{lemma} \label{lem:few_updates}
Consider the sequence of points $\tau_1, \dots, \tau_{\left(16n/\eps\right)^2}$ given by the z-curve on the grid~$\mathcal{G}$, and fix $i,j$. Using only $\Oh(\frac{1}{\eps^2} \log \frac{n}{\eps})$ updates to the node weight $w_{i,j}$, we can maintain $w_{i,j}$ as a $(1+\eps)$-approximation of the distance $\|\pi_i - (\sigma_j + \tau_k)\|$ while iterating over $k =1,\ldots, (16n/\eps)^2$.
\end{lemma}
\begin{proof}
Note that for $p := \pi_i - \sigma_j$ we have $\|p-\tau\| = \|\pi_i - (\sigma_j + \tau)\|$, and thus the distance that we want to maintain is $\|p-\tau_k\|$. 
Additionally, note that $p$ lies on the z-curve as the coordinates of the curves are rounded to $\frac{\delta}{4n}\eps$ and, furthermore, if $p$ is further than $\delta$ from $\taustart = \pi_1 - \sigma_1$, then using the pair $i,j$ in the traversal would incur a distance of more than $\delta$ for all $q \in Q$ and thus we can just set $w_{i,j} = \infty$.

Now, consider the following process on the recursive definition of the z-curve: Starting with the root $C_{0,0}$, recursively explore all children of each cell. Stop the process at a cell $C$ if there is an $\ell \in \ZZ$ such that all points $\tau \in C$ have a distance
\[
	(1+\eps)^{\ell-1} \leq \|p - \tau\| \leq (1+\eps)^{\ell+1}.
\]
In this case, for all points in $C$ the value $(1+\eps)^\ell$ is a valid $(1+\eps)$-approximation for the distance to $p$, so we associate the distance $(1+\eps)^\ell$ with $C$.

Note that the process is well-defined, since at the lowest level of the recursion a cell contains only a single point of $\mathcal{G}$, and thus at the latest on this level the process will stop.
The process partitions the grid $\mathcal{G}$ into cells $C_1, \dots, C_t$, which are exhaustively explored in this order by the z-curve, see Figure~\ref{fig:partition} for an illustration.
\begin{figure}
	\centering
	\includegraphics[width=.6\textwidth]{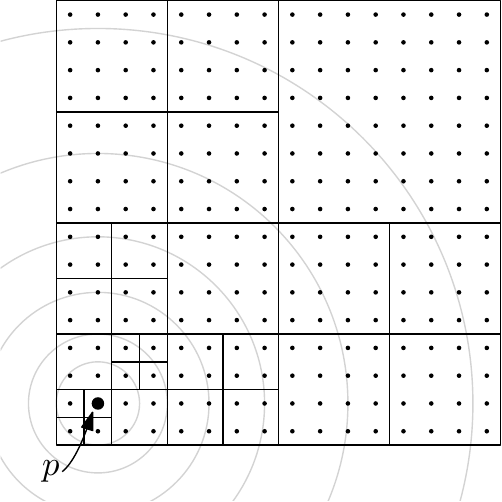}
	\caption{The partition of cells induced by the point $p$ on the recursively defined cells of the z-curve such that for each cell a single update suffices to ensure a $(1+\eps)$-approximation on the distance to $p$. Note that the cells become larger when further away from $p$.}
	\label{fig:partition}
\end{figure}
In particular, the value of $t$ is an upper bound on the number of updates needed to approximately maintain $\|p-\tau_k\|$ while iterating over all points in $\mathcal{G}$ in z-order.
We next bound the diameter of the cells for a specific associated distance, to subsequently show that this induces a small number of cells in the partition.
To this end, consider a specific cell $C_r$ and its associated distance $(1+\eps)^\ell$.
As we continued exploring the children of the parent cell $C$ of $C_r$, there have to be two points $z_1, z_2 \in C$ such that either $\|z_1 - p\| < (1+\eps)^{\ell-1} \text{ and } \|z_2 - p\| > (1+\eps)^{\ell}$, or $\|z_1 - p\| < (1+\eps)^{\ell} \text{ and } \|z_2 - p\| > (1+\eps)^{\ell+1}$.
By triangle inequality, $C$ has diameter at least
\[
	\|z_1 - z_2\| \geq \|z_2 - p\| - \|z_1 - p\| > (1+\eps)^{\ell} - (1+\eps)^{\ell-1} = (1+\eps)^{\ell-1}((1+\eps) - 1) = (1+\eps)^{\ell-1} \eps.
\]
In the recursive definition of the z-curve, the diameter of a cell decreases at most by a constant factor from parent to child if the child is not a single point.
Thus $C_j$ has diameter $\Omega((1+\eps)^\ell \eps)$ if $|C_j| > 1$.
If $|C_j| = 1$, then the Voronoi cell of $C_j$ of the Voronoi diagram of $\mathcal{G}$ has diameter $\Omega((1+\eps)^\ell \eps)$.
Thus, for both cases it holds that there is a square with area $\Omega((1+\eps)^{2\ell} \eps^2)$ that only contains points from $C_j$ but no other cell $C_{j'}, j' \neq j$.

Recall that the area of a ball of radius $R$ in the $L_\rho$-norm is equal to $\nu_\rho R^2$, where $\nu_\rho$ depends only on the $L_\rho$-norm and is thus a constant for our purpose. Hence, the area of all points between distance $(1+\eps)^{\ell-1}$ and $(1+\eps)^{\ell+1}$ from $p$ is equal to
\[
	\nu_\rho (1+\eps)^{2(\ell+1)} - \nu_\rho (1+\eps)^{2(\ell-1)} = \nu_\rho (1+\eps)^{2(\ell-1)} ((1+\eps)^4 - 1) = \Oh((1+\eps)^{2\ell}\eps).
\]
Thus, there can be at most
\[
	\frac{\Oh((1+\eps)^{2\ell} \eps)}{\Omega((1+\eps)^{2\ell} \eps^2)} = \Oh\left(\frac{1}{\eps}\right)
\]
cells associated with distance $(1+\eps)^\ell$.
Finally, there are at most $\Oh(\log_{1+\eps} \frac{n}{\eps}) = \Oh(\frac{1}{\eps} \log \frac{n}{\eps})$ different associated distances, as the minimum non-zero distance is $\Omega(\frac{\delta}{n}\eps)$ and the largest distance is $\Oh(\delta)$.
Consequently, the total number of updates can be bounded by $\Oh(\frac{1}{\eps^2} \log \frac{n}{\eps})$.
\end{proof}

We now discuss how we explicitly compute the updates. Note that explicitly checking for updates in each node of the traversal of the z-curve is prohibitive. Thus, we have to devise a non-naive way of computing the updates. Indeed, Lemma~\ref{lem:few_updates} can be turned into an algorithm.
\begin{lemma} \label{lem:explicit_updates}
The updates in Lemma~\ref{lem:few_updates} can explicitly be computed in time $\Oh\left(\frac{1}{\eps^2} \cdot \log \frac{n}{\eps}\right)$.
\end{lemma}
\begin{proof}

Lemma~\ref{lem:few_updates} already is constructive, as we associated updates to cells and thereby we can simply perform these updates at the first point of such cells.
It therefore only remains to bound the running time of all steps.
The running time for exploring the z-curve tree is dominated by the number of cells in the partition, i.e., by the number of updates, multiplied with the running time of deciding whether to explore further or not.
If the point $p = \pi_i - \sigma_j$ of Lemma~\ref{lem:few_updates} is contained in the currently considered cell, then we have to continue exploring. Otherwise, we can check if all points lie in a $(1+\eps)^{\ell-1}$ to $(1+\eps)^{\ell+1}$ distance window for any $\ell \in \ZZ$ by computing the distance to the closest and furthest point in the cell from~$p$.
All of the above steps can be done in $\Oh(1)$ time.
Finally, note that no sorting of the updates is necessary, as exploring the z-curve tree via depth-first search in the order of the z-curve already constructs the updates in sorted order.
Hence, the running time of explicitly computing the updates is dominated by the number of updates itself.
\end{proof}

We can directly use Lemma~\ref{lem:explicit_updates} to compute the updates for all node weights $w_{i,j}$. However, computing them separately would additionally incur the cost of merging them into a sorted order. We can avoid this sorting step by constructing the updates for all node weights $w_{i,j}$ in parallel using a single DFS on the z-order tree.
During the DFS, we maintain a set $E$ of pairs $(i,j)$ for which recursing further is necessary; in the top cell, this is set to $[n] \times [m]$. Then for each cell in the DFS, we need to decide for each pair $(i,j) \in E$ whether a single weight $w_{i,j} = (1+\eps)^{\ell}$ suffices to approximate the distance in this cell, which in total takes time $\Oh(|E|)$. If for $(i,j) \in E$ this is the case, then we add an update of $w_{i,j}$ to $(1+\eps)^\ell$ for the first point $\tau$ in this cell, and remove $(i,j)$ from $E$ for the recursive calls that explore the children of this cell. (We add back $(i,j)$ after the exploration.)
This process creates the updates in order and thus we do not have to sort them in a postprocessing step. It follows that the updates for all node weights $w_{i,j}$ can explicitly be computed in time $\Oh\left(nm \frac{1}{\eps^2} \cdot \log \frac{n}{\eps}\right)$.

\paragraph*{Main theorem.}

Finally, we obtain our main theorem.

\begin{theorem} \label{thm:main}
	Assume a data structure for approximate shortest paths in a directed grid graph with $N$ vertices and fixed vertices $s,t$, supporting updates of an edge weight in time $U(N)$ and $(1+\epsilon)$-approximate $s$-$t$-distance queries in time $Q(N)$. We can $(1+\epsilon)$-approximate DTW under translation in $L_p$-norm in time $\Oh(U(nm)\frac{nm}{\epsilon^2} \log^2 \frac{n}{\epsilon} + Q(nm)\frac{m}{\epsilon^2} \log \frac{n}{\epsilon})$.
\end{theorem}
The proof is deferred to Appendix~\ref{sec:deferred_proofs}; it follows easily by combining the above arguments. 

\subsection{Solving the dynamic graph problem}

Consider the data structure assumed in Theorem~\ref{thm:main} for maintaining shortest paths in a directed grid graph.
Das et al.~\cite{DPGWN22} obtain a trade-off of update time $U(N)=\tOh(N^{r})$ and query time $Q(N) = \tOh(N^{1-r})$ even for \emph{exact} distance queries in directed planar graphs where $r \in [0,\frac{1}{2}]$ is an adjustable parameter, all updates are given in advance, and all edge-weights are non-negative (both are the case in our setting).
The aforementioned result improves bounds due to Fakcharoenphol and Rao~\cite{FakcharoenpholR06} and Klein~\cite{Klein05}, also see \cite{KaplanMNS17, GawrychowskiK18}, by considering the offline setting.
For tight conditional lower bounds for the offline setting, we refer to~\cite{AbboudD16}. 
By setting $r$ such that $N^r = \sqrt{m}$ (which satisfies $r \in [0, \frac{1}{2}]$), we obtain the following corollary.

\begin{corollary}
	We can $(1+\epsilon)$-approximate DTW under translation in $\RR^2$ under the $L_p$-norm in time
	$\tOh(nm^{1.5} / \epsilon^2)$.
\end{corollary}
Note that for $n=m$ this becomes $\tOh(n^{2.5}/\eps^2)$.
It is straightforward to generalize our algorithm to $\RR^d$ for constant $d$. To this end, we have to replace the $2$-dimensional $\eps$-grid and the $2$-dimensional space-filling curve by their $d$-dimensional counterparts and adapt the analysis accordingly. The running time then merely increases with respect to the dependency on $\eps$.
\begin{corollary}
	We can $(1+\epsilon)$-approximate DTW under translation in $\RR^d$ under the $L_p$-norm with $d \in \Oh(1)$ in time
	$\tOh(nm^{1.5} / \epsilon^d)$.
\end{corollary}

\section{Conclusion and open problems}

We give the first rigorous algorithms for Dynamic Time Warping under translation, specifically an exact $\Oh(n^{2(d+1)})$-time algorithm for the $L_1$ norm in $\mathbb{R}^d$, as well as a $(1+\eps)$-approximate $\tOh(n^{2.5}/\epsilon^2)$-time algorithm for the $L_p$-norm in $\mathbb{R}^2$.

The most interesting open problem is to determine whether under the $L_2$-norm, DTW under translation admits an $\tOh(n^2 f(1/\epsilon))$-time approximation scheme.
In fact, one might be able to improve over our $\tOh(n^{2.5}/\epsilon^2)$-time algorithm via purely graph-theoretic improvements for dynamic shortest path algorithms in grid graphs, applying Theorem~\ref{thm:main} as a black box.
Specifically, we showed how to reduce $(1+\epsilon)$-approximate DTW under translation to (approximately) maintaining the $s$-$t$ distance in a directed grid graph undergoing edge-weight updates.
Our precise bound follows from plugging in a data structure due to Das et al.~\cite{DPGWN22} that maintains all \emph{exact} distances.
In fact, compared to their setting, our target problem has several important restrictions that may help to design faster algorithms:
\begin{itemize}
	\item Instead of exact distances, our application only requires a $(1+\epsilon)$-approximation.
	\item Our restriction to directed grid graphs might turn out significantly simpler than general planar digraphs.
	\item We only ever query the distance between a single source-sink pair.
\end{itemize}
Finally, if no further algorithmic improvements can be found, can we give improved conditional hardness results, going beyond our reduction from DTW for a fixed translation?

\bibliography{main}

\appendix

\section{Deferred proofs} \label{sec:deferred_proofs}

\begin{proof}[Proof of Lemma~\ref{lem:lipschitz}]
Consider any traversal $T$ of $\pi$ and $\sigma$. The cost for this traversal without translation and with translation $\tau$ can be at most $(n+m-1)\cdot \|\tau\|$ apart as the distance between two points is $1$-Lipschitz and a traversal has at most $n+m-1$ steps. Formally, by the triangle inequality and $|T| \le n+m-1$, we have
\[
	\sum_{(i,j) \in T} \|\pi_i - (\sigma_j + \tau)\| \leq |T| \cdot \|\tau\| + \sum_{(i,j) \in T} \|\pi_i - \sigma_j\| \leq (n+m-1) \|\tau\| + \sum_{(i,j) \in T} \|\pi_i - \sigma_j\|.
\]
Symmetrically it follows that $\sum_{(i,j) \in T} \|\pi_i - (\sigma_j + \tau)\| \ge \sum_{(i,j) \in T} \|\pi_i - \sigma_j\| - (n+m-1) \|\tau\|$. Since this holds for any traversal $T$, by taking the minimum over all traversals $T$, we obtain $\dtw(\pi,\sigma) - (n+m-1)\|\tau\| \le \dtw(\pi,\sigma+\tau) \le \dtw(\pi,\sigma) + (n+m-1) \|\tau\|$.

More generally, we can show $|\dtw(\pi,\sigma + \tau') - \dtw(\pi,\sigma+\tau)| \le (n+m-1) \|\tau - \tau'\|$ for any translations $\tau,\tau'$, by applying the above argument to the curves $\hat \pi := \pi, \hat \sigma := \sigma + \tau'$ and the translation $\hat \tau := \tau - \tau'$. Thus, $\dtw(\pi,\sigma+\tau)$ is $(n+m-1)$-Lipschitz in $\tau$.
\end{proof}

\begin{proof}[Proof of Lemma~\ref{lem:2n-approx}]
	Let $\delta^* \coloneqq \dtwt(\pi, \sigma)$ and let $\tau^*$ be such that $\dtw(\pi, \sigma + \tau^*) = \delta^*$, which implies that $\norm{\pi_1 - (\sigma_1 + \tau^*)}\le \delta^*$. Thus $\norm{\taustart - \tau^*} = \norm{\pi_1 - (\sigma_1 + \tau^*)} \le \delta^*$. Together with Lemma~\ref{lem:lipschitz}, we obtain $\dtw(\pi, \sigma + \taustart) \le \dtw(\pi, \sigma + \tau^*) + (n+m-1) \delta^* = (n+m) \delta^*$.
\end{proof}

\begin{proof}[Proof of Lemma~\ref{lem:dense_cell}]
	Towards a contradiction, assume that $\dtwt(\pi, \sigma) \leq \delta$, i.e., there exists a $\tau$ with $\dtw(\pi, \sigma + \tau) \leq \delta$, but the box $B_i$ with $\tau \in B_i$ is not in the neighborhood of a dense box. Then there are less than $9 \cdot \frac{n}{18} = \frac{n}{2}$ points in $P$ which are in distance at most $\frac{2\delta}{n}$ from $\tau$. As any traversal of $\pi$ and $\sigma$ consists of at least $n$ steps, it also contains more than $\frac{n}{2}$ pairs in distance more than $\frac{2\delta}{n}$ and thus $\dtw(\pi, \sigma + \tau) > \frac{n}{2} \cdot \frac{2\delta}{n} > \delta$.
\end{proof}

\begin{proof}[Proof of Lemma~\ref{lem:correctness}]
	We prove the contrapositive.
	Assume $\dtwt(\pi, \sigma) \leq \delta$. Then it follows from Lemma~\ref{lem:dense_cell} that there exists a box $B_i \in N(B_j)$ for some dense box $B_j$ and a $\tau \in B_i$ such that $\dtw(\pi, \sigma + \tau) \leq \delta$. Thus, there exists a $q \in Q$ with $\|\tau - q\| \leq \frac{\delta}{4n}\epsilon$. As DTW is $(n+m-1)$-Lipschitz in the translation (Lemma~\ref{lem:lipschitz}), we have that
\[
	\dtw(\pi, \sigma+q) \leq (n+m-1) \frac{\delta}{4n}\epsilon + \dtw(\pi, \sigma+\tau) \leq \frac{\eps}{2}\delta + \delta \leq (1+\frac{\eps}{2})\delta,
\]
which concludes the proof.
\end{proof}

\begin{proof}[Proof of Theorem~\ref{thm:main}]
First, we argue that for an approximate decision we need time
\[
	\Oh\Big(U(nm)\frac{nm}{\epsilon^2} \log \frac{n}{\epsilon} + Q(nm)\frac{m}{\epsilon^2}\Big).
\]
Note that the size of the DTW graph is $nm$, thus $N = nm$.
By Lemma~\ref{lem:correctness} the number of queries is the size of the set $Q$, which is $\Oh(\frac{m}{\epsilon^2})$.
By Lemma~\ref{lem:explicit_updates} the total number of updates and also the time to construct them is $\Oh(\frac{nm}{\epsilon^2} \log \frac{n}{\epsilon})$.

To compute the value of DTW under translation up to a factor of $(1+\eps)$, we perform binary search as follows. First, we compute a $(n+m)$-approximation by $\delta \coloneqq \dtw(\pi, \sigma+\taustart)$ with $\taustart = \pi_1 - \sigma_1$, see Lemma~\ref{lem:2n-approx}. Thus, the optimal solution has to lie in the range $[\delta/(n+m), \delta]$. Then, we perform binary search on this range using our approximate decider with an approximation parameter of $\eps/3$. It takes $\Oh(\log \frac{n}{\eps})$ steps until we narrow down the range to an interval of the form $[\delta', (1+\eps)\delta']$, which yields a $(1+\eps)$-approximation. The claimed number of updates and queries follows.
\end{proof}

\section{\boldmath DTW under translation in $L_1$} \label{sec:exact_l1}

In this section we present an exact polynomial-time algorithm for DTW under translation in~$\RR^d$ under $L_1$, and argue that it is near-optimal for a restricted class of curves.
Note that in $\RR^2$ we can transform an $L_\infty$ instance into an $L_1$ instance by rotating it by $\frac{\pi}{2}$ and scaling it by $\frac{1}{\sqrt{2}}$. This trick does not carry over to $d > 2$.

\subsection{\boldmath $\Oh(n^{2(d+1)})$ Algorithm for $L_1$}

We first describe a structural insight that subsequently helps us to give an algorithm with the claimed running time.

\begin{lemma} \label{lem:l1_structure}
Let $\pi = (\pi_1, \dots, \pi_n)$ and $\sigma = (\sigma_1, \dots, \sigma_m)$ be two curves in $\RR^d$ under the $L_1$ norm, and let $f(\tau) \coloneqq \dtw(\pi, \sigma+\tau)$.
Some minimizer $\tau^* \in \arg\min f(\tau)$ is contained in
\[
	\{ (\tau_1, \dots, \tau_d) \mid \tau_k \in \{ \pi_i[k] - \sigma_j[k] \mid i \in [n], j \in [m] \} \}.
\]
\end{lemma}
\begin{proof}
Consider the set
\[
	P \coloneqq \{ \pi_i - \sigma_j \mid i \in [n], j \in [m] \},
\]
which consists of all translations that translate any point $\sigma_j$ onto any point $\pi_i$. Consider the grid resulting from all axis parallel planes through points in $P$. This grid induces a partition $\mathcal{P}$ of $\RR^d$ (where the cells of the partition overlap at their borders).
In the following we show that inside each cell the minimum of $f$ is attained at a corner of the cell. This implies the lemma, since the set $\{ (\tau_1, \dots, \tau_d) \mid \tau_k \in \{ \pi_i[k] - \sigma_j[k] \mid i \in [n], j \in [m] \} \}$ precisely describes all corners of cells of our partition.

Note that for any point $p \in \RR^d$, the distance function $\norm{p - q}_1$ is a piece-wise linear function over $q \in \RR^d$ consisting of $2^d$ pieces. These pieces form a partition that is induced by all $d$ planes that go through $p$ and are orthogonal to one of the coordinate axes.
Now consider a cell $[x_1, x_1'] \times \dots \times [x_d, x_d']$ of the partition $\mathcal{P}$. As the cells are defined by axis-parallel planes through the points of $P$, there cannot be any point with $i$th coordinate in the range $(x_i, x_i')$. Thus, for all $\pi_i$ and $\sigma_j$, the distance $\norm{\pi_i-(\sigma_j+\tau)}$ is a linear function inside of the cell.
It follows that within a cell the Dynamic Time Warping distance for a fixed traversal is linear as a function of $\tau$, since it is a sum of linear functions.
Minimizing over all traversals, we see that within each cell the function $f(\tau)$ is a minimum of linear functions and thus concave. Consequently, within each cell $f(\tau)$ attains its minimum at a corner of the cell.
\end{proof}
The algorithm to compute the minimum then follows immediately.
\begin{theorem}
	Dynamic Time Warping under translation in $\RR^d$ under $L_1$ for two curves of length $n$ and $m$ can be computed in time $\Oh((nm)^{d+1})$.
\end{theorem}
\begin{proof}
We first construct $P = \{\pi_i - \sigma_j \mid i \in [n], j \in [m]\}$ in $\Oh(nm)$ time. Then we construct the grid by sorting according to each dimension, respectively, in time $\Oh(nm \log nm)$. As the minimum has to lie in a corner, we can simply compute $\dtw(\pi,\sigma+\tau)$ for all the corners $\tau \in \RR^d$ of the grid. We then return the minimum. As there are $\Oh((nm)^d)$ corners and computing the Dynamic Time Warping distance is in time $\Oh(nm)$, we end up with a running time of $\Oh((nm)^{d+1})$.
\end{proof}

\subsection{Curves with a constant number of different coordinates}

Consider the restricted class of curves in which all coordinates take values in $C$, for a fixed set $C$ of constant size. The DTW problem remains non-trivial with this restriction. This class of curves comes up in the SETH-hardness proof of DTW~\cite{BringmannK15}.

Note that when all coordinates come from $C$ we can further improve the running time bound of our algorithm from the previous section. Indeed, then each candidate translation must have coordinates in $C - C = \{x - y \mid x,y \in C\}$. Since $|C-C| \le |C|^2$, there are at most $|C|^{2d}$ candidate translations. 
In particular, the overall running time of the algorithm becomes $\Oh(|C|^{2d} nm)$, which is $\Oh(nm)$ for constant $|C|$ and $d$.

When considering this tighter bound on the running time, then our algorithm is tight under SETH in any constant dimension $d$.
Indeed, it is known that computing DTW with a fixed translation requires time $(nm)^{1-o(1)}$ assuming the Strong Exponential Time Hypothesis, even in dimension $d=1$~\cite{BringmannK15}. Using our arguments from Appendix~\ref{sec:LBfixedtranslation}, this hardness carries over to DTW under translation (using a single additional coordinate value). By embedding $d=1$ into higher dimensions, the same holds for any constant dimension $d$.
Hence, there is a constant-size set $C$ such that for any constant dimension $d$ DTW under translation for curves with coordinates in $C$ requires time $(nm)^{1-o(1)}$. This matches our running time of $\Oh(nm)$ for constant $C$ and $d$.

\section{Lower bound} \label{sec:LBfixedtranslation}

Intuitively, DTW under translation is at least as hard as DTW for a fixed translation. 
The following proposition formalizes this intuition, both for exact and approximation algorithms.

\begin{proposition}\label{prop:LB}
	Let $c\ge 1$. If DTW under translation (under $L_p$-norm) can be $c$-approximated in time $T(n)$, then DTW (under $L_p$-norm) can be $c$-approximated in time $O(T(n))$.
\end{proposition}
From this proposition, we obtain the following consequences:
\begin{itemize}
	\item \emph{Exact algorithms:} We cannot solve DTW under translation in time $O(n^{2-\delta})$ with $\delta > 0$ unless the Strong Exponential Time Hypothesis is false. This follows from the corresponding result for DTW for fixed translation which holds already for $d=1$~\cite{AbboudBW15,BringmannK15}.
	\item \emph{Approximation algorithms:} A $(1+\epsilon)$-approximation algorithm for DTW under translation in time $O(n^{2-\delta}f(1/\epsilon))$ with $\delta > 0$ and any function $f$ would give a novel approximation algorithm for DTW for fixed translation, beating current guarantees, cf.~\cite{Kuszmaul19, AgarwalFPY16, YingPFA16}. 
\end{itemize}

To state the proof, let us prepare some notation: For two curves $\rho=(\rho_1, \dots, \rho_n), \rho'=(\rho_1',\dots,\rho_m')$, we let $\rho \circ \rho' = (\rho_1, \dots, \rho_n, \rho_1',\dots,\rho_m')$ denote the concatenation of $\rho$ and $\rho'$. Furthermore, for a point $r\in \mathbb{R}^2$ and $n\in \mathbb{N}$, let $r^n$ denote the $n$-vertex curve given by the sequence of $n$ copies of $r$. 

\begin{proof}[Proof of Proposition~\ref{prop:LB}]
	Let $\pi, \sigma$ be polygonal curves with at most $n$ vertices and let $b > 0$ be such that $\pi,\sigma$ are contained in the radius-$b$ ball $B\coloneqq B_b(0)$, centered around the origin. Let $\pi', \sigma'$ be obtained from $\pi, \sigma$ by prepending $2n$ copies of the point $r \coloneqq (5 nb, 0)$, i.e., $\pi' = r^{2n} \circ \pi$ and $\sigma' = r^{2n} \circ \sigma$.  We prove that $\dtwt(\pi',\sigma') = \dtw(\pi, \sigma)$, which yields the desired reduction.

	We first show that $\dtwt(\pi',\sigma') \le \dtw(\pi, \sigma) \le 2nb$: Consider the translation $\tau = 0$ and any traversal of $r^{2n},r^{2n}$ -- this incurs a cost of 0 -- followed by an optimal traversal of $\pi,\sigma$ -- this incurs a cost of $\dtw(\pi, \sigma)$. We thus obtain $\dtwt(\pi',\sigma')\le \dtw(\pi, \sigma)$. Furthermore, since $\pi,\sigma$ are contained in $B$, any traversal with $L$ steps has cost at most $2Lb$, and the second inequality follows, since there is a traversal of $\pi, \sigma$ with at most $n$ steps.  
	
	It remains to prove that $\dtwt(\pi',\sigma') \ge \dtw(\pi, \sigma)$: Let $\tau^*$ be such that $\dtw(\pi', \sigma' + \tau^*)$ is minimized. If $\norm{\tau^*} > b$, then the first $2n$ pairs of any traversal of $\pi',\sigma'+\tau^*$ incur a cost of $\norm{r - (r+\tau^*)}=\norm{\tau^*} > b$ each. Thus $\dtw(\pi', \sigma' + \tau^*) > 2nb$, which is a contradiction. Thus, we have $\norm{\tau^*} \le b$. The optimal traversal of $\pi',\sigma'+\tau^*$ must be a traversal of $r^{2n},r^{2n}+\tau^*$, followed by a traversal of $\pi,\sigma + \tau^*$, since otherwise the traversal contains a pair of the form $r, \sigma_i+\tau^*$ or $\pi_j, r+\tau^*$, which incurs a cost of at least $\norm{r} - (\max_{x\in B} \norm{x}) - \norm{\tau^*} \ge 5nb - 2b > 2nb$. Since an optimal traversal of $r^{2n},r^{2n}+\tau^*$ incurs cost $2n\norm{\tau^*}$, we have that
	\[\dtw(\pi', \sigma'+\tau^*) \ge 2n\norm{\tau^*} + \dtw(\pi, \sigma+\tau^*) \ge 2n\norm{\tau^*} + \dtw(\pi, \sigma) - 2n\norm{\tau^*} = \dtw(\pi, \sigma),\]
	where we used the Lipschitz property (Lemma~\ref{lem:lipschitz}) in the second inequality. This concludes the proof of $\dtwt(\pi',\sigma')=\dtw(\pi,\sigma)$.

	Finally, note that $\pi', \sigma'$ are polygonal curves with $O(n)$ vertices that can be constructed in time $O(n)$. Thus, any algorithm $c$-approximating $\dtwt(\pi', \sigma')$ in time $T(n)$ yields an algorithm $c$-approximating $\dtw(\pi, \sigma)$ in time $O(T(n))$, since $\dtwt(\pi', \sigma') = \dtw(\pi, \sigma)$.
\end{proof}

\end{document}